\documentclass[11pt,letterpaper]{article}

\usepackage{fullpage}
\usepackage{amsmath,amssymb,amsthm,dsfont}

\usepackage{enumerate}
\usepackage{fullpage}
\usepackage{array}

\usepackage{algorithm,algorithmic}

\usepackage{graphicx}
\usepackage[pdftex,pdfpagelabels,bookmarks,hyperindex,hyperfigures]{hyperref}
\usepackage{cleveref}
\hypersetup{colorlinks, allcolors=magenta}

\usepackage{array,multirow,makecell}
\setcellgapes{3pt}
\makegapedcells
\newcolumntype{L}[1]{>{\raggedright\arraybackslash }b{#1}}
\newcolumntype{C}[1]{>{\centering\arraybackslash }b{#1}}

\newcommand{\In}{{\mathrm {\,In}}}
\newcommand{\diam}{{\mathrm{diam}}}

\newcommand{\SAP}{SAP_{g}}

\newcommand{\FP}{SAP_{\lambda x.M}}
\newcommand{\MM}{\textit{MinMax}}

\newcommand{\IN}{\mathds{N}}

\newcommand{\dG}{\mathds{G}}

\newcommand{\MZ}{M}                %
\newcommand{\e}{{\mathrm{e}}}
\newcommand{\rad}{{\mathrm{rad}}}

\newcommand{\Roo}{{\mathrm{Roots}}}

\renewcommand{\leq}{\leqslant}
\renewcommand{\geq}{\geqslant}

\newcommand{\eqdef}{\stackrel{\text{def}}{=}}

\newcommand{\V}{{\mathcal V}}

\newcommand{\Z}{Z}

\newtheorem{thm}{Theorem}

\newtheorem{lem}[thm]{Lemma}
\newtheorem{cor}[thm]{Corollary}

\newcommand{\Clmod}[1]{\big [\, {#1} \, \big ]}
\newcommand{\clmod}[2]{[ {#1} ]_{_{#2}}}
\newcommand{\nmin}[2]{ {#1}_{_{\, #2}}^{^{+}}}

\usepackage{tikz}
\usetikzlibrary{positioning}

\usepackage{caption}
\usepackage{subcaption}

\newcommand{\ignore}[1]{}

\date{\today}
\begin{document}

\title{Self-Stabilizing Clock Synchronization in Dynamic Networks 
}

 \author{%
 Bernadette Charron-Bost\textsuperscript{1} 
  \and Louis Penet de Monterno\textsuperscript{2}}
 \date{\textsuperscript{1} DI ENS, \'Ecole Normale Sup\'erieure, 75005 Paris, France\\
 \textsuperscript{2}  \'Ecole polytechnique, IP Paris, 91128 Palaiseau, France\\~\\
 \today
  }

\maketitle

\begin{abstract}
We consider the fundamental problem of clock synchronization in a synchronous multi-agent system.
Each agent holds a clock with an arbitrary initial value, and clocks must eventually indicate the same value.
Previous algorithms worked in static networks with drastic
	connectivity properties and assumed that global  information is available at each agent.
In this paper, we propose different solutions for time-varying topologies that  require neither strong connectivity nor 
	any  global  knowledge on the network.

First, we study the case of unbounded clocks, and propose a self-stabilizing \emph{MinMax} algorithm 
	 that works if, in each sufficiently long but bounded period of time, there is an agent, called a root, 
	 that can send messages, possibly indirectly, to all other agents. 
Such networks are highly dynamic in the sense that roots may change arbitrarily over time.
Moreover, the bound on the time required for achieving this rootedness property is unknown to the agents. 
Then we present a finite-state  algorithm that synchronizes periodic clocks in dynamic networks that are 
	strongly connected over bounded period of time.
Here also, the bound on the time for achieving strong connectivity exists, but  is not supposed to  be known. 
Interestingly, our algorithm  unifies several seemingly different algorithms proposed previously for static networks. 
Next, we show that 	strong connectivity is actually not required: our algorithm still works when the network 
	is just rooted 
	over bounded period of time with a set of roots that becomes stable. 
Finally, we study the time and space complexities of our algorithms, and discuss how initial timing information
	allows for more efficient solutions.
\end{abstract}

\section{Introduction}

There is a considerable interest in distributed systems consisting of multiple, potentially mobile, agents. 
This is mainly motivated  by the emergence of large scale networks, characterized by the lack of centralized
	control, the access to limited information and a time-varying connectivity. 
Control and optimization algorithms deployed in such networks should be completely distributed, 
	relying only on local observations and informations, and robust against unexpected changes in topology
	such as link or node failures.
	
A canonical problem in distributed control is the \emph{clock synchronization problem}:
In a system where each agent is equipped with a local discrete clock, the objective 
	is that all clocks eventually synchronize despite arbitrary initializations.
Clock synchronization is a fundamental problem arising in a number of applications, both in 
	engineering and natural systems.
A synchronized clock is a  basic block used in many engineering systems, e.g., in 
	the universal self-stabilizing algorithm developed by Boldi and Vigna~\cite{BV:dc:02},
	or for deploying distributed algorithms structured into synchronized \emph{phases} 
	(e.g., the \emph{Two-Phase} and \emph{Three-Phase Commit} algorithms~\cite{BHG:book:87}, 
	or many consensus algorithms~\cite{BenOr:podc:83,DLS:jacm:88,Lam:acmtcs:98,CBS:dc:09}).
Clock synchronization also corresponds to an ubiquitous phenomenon in the natural world and finds numerous 
	applications in physics and biology, e.g., the Kuramoto model for the synchronization of coupled 
	oscillators~\cite{Str:phys:00}, synchronous flashing fireflies, or else
	collective synchronization of pancreatic beta cells~\cite{Jad:cacm:12}.  

Our goal in this paper is the design of  distributed  algorithms for the clock synchronization 
	problem in a networked system of $n$ agents that operate in synchronous rounds and communicate by broadcast. 
We consider both cases where agents may have an unbounded number of states and the case of finite state 
	agents and periodic clocks. 
 The network is supposed to be uniform and anonymous, i.e., agents are identical and have no identifiers.
We consider the self-stabilization model  where the initial state of each agent is arbitrary.
In particular, agents do not have a consistent numbering of  the rounds.
Moreover, agents may use only local informations.      

The communication pattern at each round is modeled by a directed graph that may change continually 
	from one round to the next. 
In other words, we allow for time-varying communication graphs, which is important if we want to take 
	into account link failure and link creation, reconfigurable networks, or for dealing with probabilistic 
	communication models  like the rumor spreading models.	
We impose weak assumptions on the communication topology; in particular, we allow for non-bidirectional links and 
	do not assume full connectivity, nor even strong connectivity. 
Indeed, the assumption of strong connectivity  may be too restrictive in various settings:
	for instance, asynchrony and benign agent failures in a fully connected network may be 
	handled by dynamic graphs that are permanently rooted, but not strongly 
	connected~\cite{CBS:dc:09}.
 
\paragraph{Contribution.}
In the case of unbounded memory,  we use the similarity between clock synchronization and consensus:\footnote{%
	Any \emph{averaging algorithm}~\cite{BT89,Mor:ieeetac:05} for consensus -- in particular, the simple \emph{Min}
	and \emph{Max} algorithms -- can be directly translated into an algorithm for clock synchronization 
	by a simple incrementing procedure.}
	$\!\!$we adopt the \emph{MinMax} algorithmic scheme developed in~\cite{CBM:dc:21} for consensus 
	in a networked system with asynchronous starts and a time-varying topology. 
While consensus and clock synchronization  are easy to achieve when 
	the time-varying topology is infinitely often strongly connected over time,  in the sense that for every pair of agents $i$ and $j$
	 there always exists a \emph{temporal path} from $i$ to $j$, 
	there is no obvious solution in the case where only a strict 
	 subset of agents, which may vary over time, can broadcast messages to all other agents.
The \emph{MinMax} scheme has been specifically designed to cope with the lack of strong connectivity.	

In fact, the most challenging point here is the derivation of a self-stabilizing  algorithm, which works with arbitrary 
	initial states instead of only tolerating asynchronous starts.
For that, each agent builds its own view and, despite varying and non-predictive communication delays,
	it succeeds in eliminating the ``garbage'' introduced by the arbitrariness of initial states.
We prove that our self-stabilizing algorithm works under a very weak connectivity assumption: 
	 the network is just required to be \emph{rooted with bounded delay}, i.e., the communication graph 
	 over every period of bounded length has a rooted spanning tree. 
This assumption corresponds to highly dynamic networks since roots and temporal paths may permanently change over time.
Moreover, the bound for guaranteeing rootedness is supposed to exist but is unknown, 
	in the sense  our algorithm does not depend on it.
		
In the second part of the paper, we impose the additional constraint of state finiteness. 
This model provides a good approximation for  networked systems that are subject 
	to communication bandwidth and storage limits. 
We present an algorithm, called \emph{SAP} (for self-adaptive period), and show that it
	solves the mod-$P$ synchronization problem in any time-varying topology that is 
	\emph{strongly connected with bounded delay}, i.e., from any time onward and for every pair of agents $i$ and $j$,
	 there is a temporal path of bounded length connecting $i$ to $j$. 
In other words,  the  \emph{SAP} algorithm works under the sole condition of a finite \emph{dynamic} diameter.\footnote{%
	Observe that the diameter of a static strongly connected network is less than the number of agents,
	while it may be arbitrarily large for a dynamic network.
	This is why the assumption of a bound on the diameter available at each agent may be quite problematic
	in the dynamic setting.}
Its stabilization time is bounded above by  three times the diameter when a bound on the diameter is given, 
	but the algorithm fundamentally  works when no bound is available, with a limited increase of stabilization time. 

Interestingly, the   \emph{SAP} algorithm unifies several  seemingly different algorithms for the synchronization 
	of  periodic clocks in static networks, including  the algorithms in~\cite{ADG:ppl:91,HG:ipl:95,BPV:algorithmica:08} 
	and the one deployed in the finite-state universal self-stabilizing protocol in~\cite{BV:dc:02},
	with useful insights for improving their solvability powers.
In particular, we show that the pioneer algorithm proposed by Arora et al.~\cite{ADG:ppl:91} works for a period~$P \geq 6n$ 
	while the authors proved its correctness only when $P \geq n^2$.

Finally, we study how to relax the strong connectivity assumption, and consider the property
	of \emph{uniform rootedness with bounded delay}:
	the network is rooted with bounded delay and the set of roots is fixed, at least from some time onward.
Neither the  bound on the delay nor the set of root agents are supposed to be known.	
We prove that the  \emph{SAP} algorithm still works under this weaker connectivity assumption thanks to
	a synchronization mechanism  quite different from the one involved in strongly connected networks.
	
\paragraph{Related work} 
Self-stabilizing clocks have been extensively studied in different communication models, under different assumptions,
	and with various problem specifications.
 The pioneer papers by Even and Rajsbaum~\cite{ER:mst:95} and by Gouda and Herman~\cite{GH:ipl:90} 
 	use the \emph{Min} and  \emph{Max} algorithms 
	for the synchronization of unbounded clocks in a fixed strongly connected communication graph.
For periodic clocks, the most closely related pieces of work are the  series of papers~\cite{ADG:ppl:91,HG:ipl:95,BPV:algorithmica:08}
	which, in addition to strong connectivity, all assume that  a bound on the diameter is available.
To the best of our knowledge, only the synchronization algorithm in~\cite{BV:dc:02} for a static communication graph 
	dispenses with the latter assumption.
	
More recently, clock synchronization has been studied in the \emph{Beeping model}~\cite{CK:disc:10} in which
	agents have severely limited communication capabilities:
	given a connected bidirectional communication graph,
	in each round, each agent can either send a ``beep'' to all its neighbors or stay silent.
A self-stabilizing algorithm for clock synchronization has been proposed by Feldmann et al.~\cite{FKS:spaa:20},
	which is optimal both in time and space, but which, unfortunately, requires that a bound on the 
	network size is available for each agent.\footnote{%
	In~\cite{FKS:spaa:20}, Feldmann et al. also proposed an algorithm that does not use any bound on the network size,
	but that only tolerates asynchronous starts, giving another hint that 
	the self-stabilization model is less powerful than the model of  asynchronous starts.
	}
	
There are also numerous results for clock synchronization with faulty agents. 
The fault-tolerant solutions that have been proposed in various failure models,
	including  the  Byzantine failure model, 
	all use algorithmic schemes initially developed for consensus (e.g., see~\cite{Dol:97,DW:jacm:04}).  
They all require a bidirectional connected (most of the time fully connected) network, and the only results 
	for unidirectional networks are about rings.

Clock synchronization has also been studied in the model of \emph{population protocols}~\cite{AAER:dc:07},
	consisting of a set of agents,  interacting  in randomly  chosen  pairs.
In this model, the underlying network is assumed to be fully connected, and the pairwise interactions
	are modeled by bidirectional links.	
Moreover, only stabilization with probability one or with high probability is required.
The same  weakening of problem specification is considered for another popular probabilistic communication model, 
	namely the \emph{PULL} model~\cite{KDG:focs:03},
	 in which, at  each round each agent interacts with one random incoming neighbor in a fixed directed graph~$G$.
Unfortunately,  the self-stabilizing clock synchronization algorithms developed in this model~\cite{BKN:dc:19,BGS:soda:21}
	assume that $G$ is  the complete graph, a non-realistic hypothesis in many settings, in particular for natural systems. 
				
\section{Preliminaries}\label{sec:model}
 
\subsection{The computing model}

We consider a networked system with a fixed and finite set $V$ of agents. 
We assume a round-based computational model in the spirit of the Heard-Of model~\cite{CBS:dc:09}. 
Point-to-point communications are organized into \emph{synchronized rounds}: each node sends messages 
	to all nodes and can receive messages sent  by \emph{some} of the nodes.
Rounds are communication closed in the sense that no node receives messages in round~$t$ that are sent 
	in a round different from~$t$.
Communication  at each round $t$ is thus modeled by a directed graph (digraph) 
	$\dG(t)=(V,E_t)$: $(i,j)\in E_t$ if and only if communication from $i$ to $j$ is enabled at round $t$. 
We assume a self-loop at each node in all these digraphs  since a node can communicate with 
	itself instantaneously.	
The sequence of digraphs~$\dG=\left (\dG(t) \right )_{t \geq 1}$ is called a {\em dynamic graph}. 

An \emph{algorithm}~${\cal A}$ is given by a set ${\cal Q}$ of local states, a set of messages ${\cal M}$,
	a sending function  $\sigma : {\cal Q}  \rightarrow {\cal M}$, and a transition function 
	$\delta : {\cal Q} \times {\cal M}^{\oplus} \rightarrow {\cal Q}$,
	where ${\cal M}^{\oplus} $ is the set of finite multisets over ${\cal M}$.
	
In this paper, we consider the \emph{self-stabilization} model, where all the nodes start to run
	the algorithm synchronously at round one, but their initial states are arbitrary in the set~${\cal Q}$.
An \emph{execution of}~${\cal A}$  with the dynamic graph $\dG$ then proceeds as follows:
In  round~$t$ $(t = 1,2\dots)$, every node applies the sending function~$\sigma$  to its current state to generate 
	the message to be sent to all nodes,
	then it receives the messages sent by its incoming neighbors in the digraph~$\dG(t)$, and finally 
	applies the transition function $\delta$ to its current state and the list of messages it has just received
	to go to  a next state. 
An execution of~${\cal A }$ is thus entirely determined by the collection of the initial states and the dynamic graph~$\dG$.
In the rest of the paper, we adopt the following notation:
	given an execution of~${\cal A }$, the value of any variable $x_i$ local to the node~$i$  at the end of 
	round~$t$ is denoted by $x_i(t)$, and $x_i(0)$ is the initial value of $x_i$ in this execution. 

The self-stabilization model is less powerful than the model of  asynchronous starts~\cite{CBM:tcs:19}.
Indeed, regarding eventual convergence properties, every self-stabilizing 
	algorithm obviously tolerates asynchronous starts.
In contrast, a node cannot measure the time elapsed since it started the computation 
	in the self-stabilization model, while it can easily do it in the second model with asynchronous 
	starts.	

 \subsection{Dynamic graphs}\label{sec:dyngraph}
 \paragraph{Graph intervals.}
Let us first  recall that the {\em product} of two digraphs $G_1=(V,E_1)$ and $G_2=(V,E_2)$, denoted $G_1 \circ G_2 $, 
	is the digraph  with the set of nodes $V$ and with an edge $(i,j)$ if there exists~$k\in V$ such that $(i,k)\in E_1$ 
	and $ (k,j) \in E_2$.
For any dynamic graph $\dG$ and any integers $t' \geq t \geq 1$, we let 
	$$ \dG(t:t') \eqdef \dG(t) \circ  \dots \circ \dG(t') .$$
By convention, $\dG(t:t)=\dG(t) $, and when $0 < t' < t$,
	$ \dG(t:t')$ is the digraph with only a self-loop at each node.

Given any dynamic graph $\dG$,  the set of $i$'s in-neighbors 
	in $\dG(t:t')$ is denoted by $\In_i(t  :t' )$,
	and simply by $\In_i(t)$ when $t' = t$.
Observe that due to self-loops, all these sets  contain the node~$i$, and
	$$ \bigcup_{t \leq \, s \, \leq t'} \In_i (s)\subseteq  \In_i (t:t')  , $$
	but the inclusion may be strict.
	
Every edge $(i,j)$ in $\dG(t:t')$ corresponds to a \emph{path in the round interval}~$[t,t']$:
	there exist $t'-t+2$ nodes $ i = k_0 , k_1,\dots ,k_{t'-t +1} = j $ such that 
	$(k_r,k_{r+1})$ is an edge of $\dG(t + r)$ for each $r = 0,\dots , t'-t $.
	
\paragraph{Center, eccentricity, and radius.}
Let us first  recall that a node~$i$ is  {\em  a root of} the digraph $G  $  if for every node  $j \in V$, 
	there is  a path from $i$ to $j$ in $G$.
The set of~$G$'s roots is denoted by $\Roo(G)$.

The \emph{eccentricity of a node}~$i$ in the dynamic graph~$\dG$, denoted~$\e_{\dG}(i)$, 
	is defined as 
	$$ \e_{\dG}(i) \eqdef \inf \{ d \in \IN^+ \,  \mid \,  \forall t\in \IN^+, \forall j \in V :  (i,j) \mbox{ is an edge in } \dG(t:t+d-1) \}.$$
The node~$i$ is \emph{central in}  $\dG$ if its eccentricity is finite, and the \emph{center} of~$\dG$, 
	denoted by $\Z(\dG)$, is defined as the set of $\dG$'s central nodes.
The \emph{diameter} of~$\dG$ is then defined as the supremum of eccentricities:
	$$ \diam(\dG) \eqdef  \sup_{i\in V}~\e_{\dG}(i) $$
	and the \emph{radius} of $\dG$ is  the infimum of eccentricities:
	$$ \rad(\dG) \eqdef  \inf_{i\in V}~\e_{\dG}(i). $$
In particular, $\diam(\dG) $ is finite if and only if  $\Z(\dG) = V$,
	 while $\rad(\dG) $ is finite if and only if $\Z(\dG) \neq \emptyset$.

We also introduce the {\em  kernel} of a dynamic graph $\dG$, denoted by $K(\dG)$ and  defined as
	$$ K (\dG) \eqdef   
		\left \{ i \in V \,  \mid \, \forall t\in \IN^+, \forall j \in V,  \exists \, t'\geq t :  
			(i,j)  \mbox{ is an edge in }  \dG(t:t') \right\} .$$
Clearly, it holds that $Z(\dG) \subseteq K(\dG) $, and the inclusion is strict in general.

\paragraph{Connectivity with bounded delay.}
Let $\Delta$ be a positive integer.
The dynamic graph $\dG$ is said to be \emph{rooted with delay}~$\Delta$ if the digraph
	$\dG(t:t + \Delta-1)$ is rooted for each positive integer~$t$.
Observe that if $\dG$ is rooted  with delay~$\Delta$, then
	it is rooted with any delay~$\Delta' \geq \Delta$ since there is a permanent self-loop at each node.
It is  \emph{rooted with bounded delay} if it is rooted  with some delay~$\Delta$.	
	
		
The dynamic graph $\dG$ is \emph{uniformly rooted with delay}~$\Delta$ if the sets $\Roo (\dG(t : t + \Delta-1))$ 
	are all equal and non-empty, and this common set of roots is then denoted by $\Z_{\Delta}(\dG)$.
We easily check that  every node in $\Z_{\Delta}(\dG)$ is central in~$\dG$,
	with an eccentricity at most equal to $ \Delta (|V| -1)$.
Moreover, if $\dG$ is uniformly rooted with delay~$\Delta$, the set~$\Z_{\Delta}(\dG)$ has no incoming edge in each 
	digraph $\dG(t)$, which shows that $\dG$ is  uniformly rooted with any delay~$\Delta'\geq \Delta$.
Moreover, $\Z_{\Delta'}(\dG) = \Z_{\Delta}(\dG)$, and this non-empty set coincides with~$\dG$'s center, namely $\Z(\dG)$. 
The dynamic graph $\dG$ is \emph{uniformly rooted with bounded delay} 
	if it is uniformly rooted with some delay~$\Delta$.
In particular, $\dG$ is \emph{strongly connected with bounded delay} 
	if it is uniformly rooted with bounded delay and $\Z(\dG) =V$, 
	which is equivalent to  just $\Z(\dG) =V$.

For every  property $\Phi$ on dynamic graphs, the self-stabilization model naturally leads to
	consider the class of dynamic graphs with a suffix satisfying $\Phi$, i.e., that ultimately satisfy $\Phi$. 
Using large enough delays and because of the permanent self-loops, this extension is actually pointless 
	for all the above   mentioned ``properties with bounded delay'', except the property of uniform rootedness 
	with bounded delay.
	
\section{Unbounded  \emph{MinMax} Clocks} \label{sec:minmax}

In this section, we present a self-stabilizing algorithm that builds unbounded  synchronized clocks 
	in  any dynamic  graph that is rooted with bounded delay.
Our algorithm uses the \emph{MinMax}  algorithmic scheme developed by Charron-Bost and Moran~\cite{CBM:dc:21}
	for coping with asynchronous starts and non-strong connectivity in the context of stabilizing consensus.
With a simple incrementing procedure, the translation into an algorithm for clock synchronization is direct;
	the more delicate point  is the derivation of a self-stabilizing algorithm, which works with arbitrary initial states 
	instead of only tolerating asynchronous starts.	
\subsection{Kernel of a rooted dynamic graph}

The \emph{MinMax} scheme relies on a fundamental property of the kernel 
	in a dynamic graph~$\dG$ that is rooted with delay $\Delta$: while a node in the kernel 
	may be non-central, i.e., may have an infinite eccentricity,\footnote{%
	For instance, if $S_i$ and $S_j$ denote the two stars centered at $i$ and $j$, respectively,
	then the dynamic graph 
	$$S_i,S_j,S_i ,S_i , S_j , S_j ,  S_i , S_i , S_i , S_j , S_j , S_j ,  \cdots $$
	is rooted with delay one, but has an infinite radius.}
	the kernel  has a finite ``global eccentricity''.
More precisely, above a certain index, there exists a dynamic path  
	of length $ \Delta (|V| - |K(\dG)|) $ from the kernel (considered as a whole) to every node in the graph.

\begin{lem}\label{lem:kernel}
If $\dG$ is rooted with delay $\Delta$, then there exists a positive integer~$s_{0}$
	such that:
	$$ \forall t \geq s_0, \forall i \in V, \ \ \In_i ( t : t + \Delta (|V| - | K(\dG)|) ) \cap K(\dG) \neq \emptyset. $$
\end{lem}

\begin{proof} 
By considering the dynamic graph $\dG_{\Delta}$ defined by $\dG_{\Delta} (t) = \dG \big( (t-1)\Delta+ 1 : t \Delta \big) $
	that is  rooted with delay one, the proof reduces to the case $\Delta = 1$.
Let $s_0$ be a positive integer such that for all $t \geq s_0$, every edge of~$\dG(t)$  
	occurs infinitely often.
Then we have that
	\begin{equation}\label{eq:blush}
	 \forall t\geq s_0, \  \  \Roo(\dG( t ))  \subseteq K(\dG )   .
	 \end{equation}
Since $\dG$ is rooted with delay one, this implies that $K(\dG ) $ is non-empty.
	
For any non-negative integer~$d$, let us now introduce the set $S_d$ of nodes that are out-neighbors 
	of the nodes in~$K (\dG)$ in the digraph~$\dG(t : t+d-1)$.	
Hence,  $ S_{0} = K (\dG) $ and because of the self-loops, 
	$S_{d} \subseteq S_{d + 1}$.
We now show that either $ S_{d } = V $ or $S_{d} \subsetneq S_{d + 1}$.

For that, assume that there is a node $i \notin S_d$, and let $j$ be a root of the digraph $\dG( t+d +1)$;	
	hence there exists a path $\gamma$ from~$j$ to $i$ in  the digraph~$\dG( t+d +1)$.
From (\ref{eq:blush}) and the above inclusions, we derive that 
	$$ j \in \Roo(\dG( t + d + 1 ))  \subseteq  K (\dG ) \subseteq S_d  .$$
Thereby, there are two consecutive nodes~$k$ and~$\ell$ along the path $\gamma$ such that
	$k \in S_d$ and $\ell \notin S_d$.
By construction, $\ell \in S_{d + 1}$, which shows that $S_{d} \neq S_{d + 1}$.

It follows that $ S_{|V| -|K(\dG)|} = V  $, and thus 
	$\In_i \big( t : t + (|V| -|K(\dG)|) \big) \cap K( \dG) \neq \emptyset $,
		as required.
\end{proof}

Let us observe that the index $s_0$ in the above lemma depends on the preliminary period in which 
	transient edges occur which may be of arbitrary length.

\subsection{The self-stabilizing MinMax clocks}

A classical tool in the study of synchronized networks is the concept of \emph{view},
	introduced for a fixed digraph~\cite{YK:ieee:96,BSVCGS:istcs:96}, and which can be easily 
	extended to dynamic graphs: 
	a view of a node is a tree that gathers all the information that the node can obtain by exchanging 
	information with its time-varying neighbors.
	
In the first lemma of this section, we demonstrate that, regarding  specific types of variables, any view may be reduced 
	to a flatten-tree structure, namely a  \emph{set} of values associated to some nodes in the tree and 
	tagged with time labels, indicating their respective depths in the tree.
More formally, a \emph{view}  is a non-empty and finite subset $\V$ of  $ \IN \times \IN$.
If $(v, d) \in \V$, we say that $v$ \emph{is a value of the view $\V$ of depth} $d$.

Our  algorithm equips each node $i$ with an integer clock~$C_i$ (see Algorithm~\ref{algo:MinMax}).
The node~$i$ has also an auxiliary clock~$h_i$ that measures the elapsed 
	time (line~\ref{line:incr}). 
For the ease of description, we introduce a third integer clock, namely the \emph{min-clock} $c_i$, 
	even if this variable is actually redundant: in each round, 
	the clock~$c_i$	 is set to  one plus the minimum of the clock values $c_j$  that $i$ has just received.\footnote{%
	The simple min-clocks $c_i$'s  clearly  achieve eventual synchronization in the case of 
	strong connectivity.} 
The node $i$ maintains a view $\V_i$ for recording the values of the clocks~$c_j$
	it has heard of.
Any value $v$ in this view is equal to the past value  of some min-clock $c_j$ 
	augmented by the time lag correction, and $v$'s depth in this view is 
	the latest round where this equality held to the best of $i$'s knowledge.
For that, at each round, $i$ increments depths by one and appends the pair $(v,0)$ to its view
	if $v$ is $c_i$'s current value.
Then the node $i$ sets its clock $C_i$ to the maximum of the values in its current view
	  of ``good enough'' depth (line~\ref{line:C}).

\begin{algorithm}[ht]
\begin{algorithmic}[1]
\vspace{0.1cm}
\REQUIRE{}
	\STATE $h_i \in \mathds{N}$
	\STATE $\V_i \in  {\cal V} $
	\STATE $C_i \in \mathds{N}$
\ENSURE{}
	\STATE send $\V_i$ to all
	\STATE receive $\V_{j_1},\V_{j_2}, \dots $ from the  in-neighbors and let $U$ denote the union of all these views
	\STATE $\V_i  \gets \left \{ \left (1 + v , 1 + d  \right ) \, \mid \, (v, d) \in U \right \} $ \label{line:minAGE}
	\STATE $\V_i \gets \V_i \cup \{(\min \V_i [1],0) \}$ \label{line:reset}
	\STATE $h_i \leftarrow  1 + h_i $ \label{line:incr}
	\STATE $C_i \leftarrow \max \left \{ v \, \mid \, \exists \, d \leq h_i / 2 : (v, d) \in \V_i \right \}$ \label{line:C}
\end{algorithmic}
\caption{The $\MM$ algorithm for unbounded clocks}
\label{algo:MinMax}
\end{algorithm}

For the correctness proof, we fix an execution of the $\MM$ algorithm with the dynamic graph~$\dG$.
From now onwards,  we let  $c_i(t) = \min \V_i(t)[1]$.
We start with the following key lemma:

\begin{lem} \label{lem:age}
Let $i \in V$ and $v \in \IN$.
For all rounds $t\geq 1$ and all integers $d \in \{0, \dots, t -1 \}$,
	the following two statements are equivalent:
	\begin{enumerate}
		\item $(v, d) \in \V_i(t)$; 
	\item the node $i$ has an in-neighbor $j$ in $\dG(t-d+1:t)$ such that $ c_j(t-d) = v - d$.
	\end{enumerate}
\end{lem}

\begin{proof}
	For the implication $(2) \Rightarrow (1)$, assume that there exists a node  $j \in \In_i(t-d+1:t)$ such that $c_j(t-d) = v- d$.
	Then, by line \ref{line:reset}, $(v-d, 0) \in \V_j (t-d)$.
Moreover, there exists a  path in the interval $[t-d+1,t]$ that we denote $j = k_0, k_1, \dots, k_d = i$.
Because of the update rule for depths, we obtain:

	$$(v-d+1, 1) \in \V_{k_1}(t-d+1), (v-d+2, 2) \in \V_{k_2}(t-d+2),  \cdots, (v, d) \in \V_{k_d}(t),$$ 
	and the implication $(2) \Rightarrow (1)$ follows.

	The direct implication is proved by an induction on $t > 0$.
\begin{enumerate}
\item \emph{Basis:} $t=1$, and thus $d = 0$.
Assuming $(v, 0) \in \V_i(1)$, we obtain $c_i(1) = v$,
	and the node $i$ is in  $\In_i(2:1) = \{ i \}$ (cf. Section~\ref{sec:dyngraph}).
\item  \emph{Inductive step:} 
Assume that the implication $(1) \Rightarrow (2)$  holds at round $t-1$ for all integers $d \in \{0, \cdots, t-2 \}$.

If $(v, d) \in \V_i(t)$, then either $d = 0$ or $d > 0$.
In the first case, we obtain $c_i(t) = v$ and $i \in \In_i(t+1:t)$, similarly to the base case.
In the case $d > 0$,  there exists some node $k \in \In_i(t)$ whose view at round $t-1$ contains $(v-1, d-1)$.
The inductive hypothesis implies that there exists a node $j \in \In_k(t-d+1:t-1)$ such that 	
	$$c_j(t-d) = v - d.$$
Since $\mathds{G}(t-d+1:t) = \mathds{G}(t-d+1:t-1) \circ \mathds{G}(t)$, it follows that $j \in \In_i(t-d+1:t)$ as required.
\end{enumerate}	
\end{proof}

Note that Lemma~\ref{lem:age} does not hold for $d =t$.
Indeed, at round~$t$, a node has not yet had enough time to eliminate the ``garbage''
	in its view at depth $t$, induced by the arbitrariness of its initial state.  

\begin{lem} \label{lem:uniker}
There is a round $t_0$ such that for all rounds $t\geq t_0$, the following holds:
	$$  \forall i\in V, \, \forall j \in K(\dG), \  c_i(t ) \leq c_j(t) .$$
\end{lem}
\begin{proof}
Because of the self-loop at node $i$ in the digraph $\dG(t+1)$, if 
	$ v $ is a value in the view $\V_i (t) $, then  $v + 1$ is a value in  $\V_i (t + 1) $, and hence
	$$ c_i(t + 1) \leq 1 + c_i(t) .$$
It follows that the sequence $ \left(c_i(t) - t \right)_{t \in \mathds{N}}$ is non-increasing.
A simple induction shows that the sequence  is lower-bounded by zero,
	and hence converges in finite time.
Thus from some round  $r_i$ onwards, the quantity $ c_i(t) - t $ is equal to a constant integer~$c_i^0$;
	we let $r_{0} = \max_{i\in V} r_i$ and $t_0 = r_0 +1+  (|V|-|K(\dG)|)\Delta$.	
	
Let $i$ be an arbitrary node.
By Lemma~\ref{lem:kernel}, there exists  $j \in In_i(r_{0}+1:t _0) \cap K(\dG)$.
Lemma~\ref{lem:age} shows that $t_0 - r_{0} + c_j(r_{0}) $ is a value in $\V_i(t_0)$, and thus 
	$$ c_i( t _0) \leq t_0 - r_{0} + c_j(r_{0}) .$$
Since $t_0 \geq r_{0} + 1 $, we have $ c_i(t_0)  = c_i^0 + t_0 $ and $ c_j(r_{0})  = c_j^0 + r_{0} $,
	which with the above inequality gives  $  c_i^0 \leq  c_j^0 $.
It follows that for all rounds $t\geq t_0 $, it holds that 
	$$c_i(t) \leq c_j(t) ,$$
	and  thus  $c_i(t) =  c_j(t)$ if  both $i$ and $j$ are  in $K(\dG)$.	
\end{proof}

\begin{thm} \label{thm:minmax}
The \emph{MinMax} clocks  synchronize in any dynamic graph that is 
	 rooted with bounded delay.
\end{thm}
\begin{proof}
We let  $c^0 = c^0_j$ where  $j $ is any node in $K(\dG)$, and
	\begin{equation}
		t_1 \eqdef \max \{  s_0 + R , t_0 +  R + 1 , 2 (R + 1) , 2 t_0 +h(0)\}
	\end{equation}
	where $R=  \Delta (|V| -|K(\dG)| )$ and $h(0)= \max_{i\in V} h_i(0)$. 
We are going to prove that for every node~$i$ and every round $t \geq t_1$,  it holds  that
	$$ \max \left \{ v \, \mid \, \exists \, d \leq h_i (t) / 2 :  (v, d) \in \V_i (t) \right \} = c^0 +t  ,$$
	which demonstrates that all the clocks $C_i$ are synchronized from  round $t_1$.
 
Since  $ t  - R \geq  s_0 $, Lemma~\ref{lem:kernel} proves  there exists a
	node~$j $ in $ K(\dG) \cap  \In_i( t-R : t)$.
The inequality 	 $ t -R - 1  \geq t_0$ gives
		\begin{equation*}
		c_j(t-R-1) = c^0 + t - R -1 ,
	\end{equation*}
		and Lemma~\ref{lem:age} shows that the tuple $(c^0 + t, R + 1)$ is in $i$'s view at round~$t$.
Moreover, we easily check that
		\begin{equation*}
		 R + 1 \leq \frac{t_1}{2} \leq  \frac{t}{2} \leq \frac{t+h_i(0)}{2} = \frac{h_i(t)}{2} 
	\end{equation*}
		which yields
		\begin{equation*}
		c^0 + t \leq \max \left \{ v \, \mid \, \exists d \leq h_i / 2, (v, d) \in \V_i \right \}
	\end{equation*}
		
Now, let $\overline{v}$ be the maximum value in $i$'s view at round $t$ whose
	depth $d$ is less or equal to~$\frac{h_i(t)}{2}$.
Then we have
		\begin{equation*}
		 t -d \geq t - \frac{h_i(t)}{2}  = \frac{t-h_i(0)}{2} \geq t_0 
	\end{equation*}
	Since  $t -1 \geq d$, Lemma~\ref{lem:age} applies and there exists a node~$j$ in  $\In_i(t-d+1:t)$ such that
		\begin{equation*}
		c_j(t-d) = \overline{v} - d 
	\end{equation*}
	Then we obtain
		\begin{equation*}
		  \overline{v}   =   c_j^0 + t  \leq c^0 + t ,
	\end{equation*}
		which completes the proof.
The equality is because $t - d \geq t_0$,
	and the inequality is by Lemma~\ref{lem:uniker}.
\end{proof}

\subsection{Clock synchronization and link failures}
When computations are organized into synchronous rounds, benign communication failures 
	are quantified by the number of message losses per round. 
Following the approach developed in the Heard-Of model~\cite{CBS:dc:09} for benign failures,
	message losses in a static network are handled by a fault-free dynamic network with a fixed set of agents
	and time-varying links. 
In~\cite{CBFN:icalp:15}, Charron- Bost et al. showed that any digraph with $n$ nodes and at least $n^2-3n +3$ edges is rooted. 
Taking into account the $n$ self-loops and since $n^2 -3n+3 = (n^2-n)-(2n-3)$, we derive the following solvability result from 
	Theorem~\ref{thm:minmax}.
\begin{cor}
Self-stabilizing clock synchronization  may be achieved in a fully connected network with $n$ agents 
	if there are at most $2n-3$  message losses per round.
\end{cor}

\section{Periodic Clocks with a Finite Diameter}

Synchronized non-wrapping clocks require infinite states, and 
	if we require state finiteness, then every clock must be eventually periodic.
In other words, the relevant clock specification in the finite state framework is:
	$$\exists t_0, \forall t\geq t_0, \forall i ,j \in V,  \ C_i(t) \equiv_P C_j(t) $$
	for some fixed positive integer period $P$.
Even in the case of a a static strongly connected network, the naive algorithm consisting in the
	update rule:
	$$ C_i \leftarrow \Clmod{\underset{j \in \In_i}{\min}~C_j + 1}_{P} ,$$
	where $\clmod{c}{P}$ denotes the remainder of the Euclidean division of $c$  by~$P$,
	does not work  when the network diameter is too large compared to the period~$P$.
To overcome this problem, we present an algorithm, called \emph{SAP} (for self-adaptive period),
	largely inspired by the ideas developed by Boldi and Vigna~\cite{BV:dc:02} for their finite-state universal self-stabilizing
	algorithm in a static strongly connected network.
The key point of the \emph{SAP} algorithm  lies in the fact that for any positive integer~$M$, we have
	$$\left [\,  [\, c \, ]_{_{PM}} \, \right]_{_{P}} = [\, c \, ]_{_P} .$$
More precisely, each node~$i$ uses an integer variable $M_i$ and computes the clock value $C_i$
	not modulo~$P$, but rather modulo the time-varying period~$PM_i$.
The variable $M_i$ is used as a guess to find a large enough multiple of~$P$ so to make 
	the clocks eventually stabilized.
Until synchronization, the variables $M_i$ increase so that there is 
	``enough space'' between the largest clock value and the shortest period $PM_i$ in the network.
The algorithm is  parametrized by a non-decreasing function\footnote{%
	For the sake of simplicity, all the nodes use the same function, but the variant of the algorithm
	with a function~$g_i$ for each node~$i$ may be equally considered.}
	$g : \IN \rightarrow \IN$, and the corresponding algorithm will be denoted $\SAP$.

Let $g : \IN \rightarrow \IN$ be a non-decreasing  function.
If $q$ is a positive integer, $g^q$ denotes the $q$-th iterate of $g$,
	and $g^0$ is the identity function.
For every non-negative integer $m$, we let 
	$$ g^*(m) \eqdef \inf \{ q \in \IN \,  \mid \, g^q(0) \geq m \}.$$
The choice of $g$ may follow one of the two typical strategies below.
\begin{enumerate}
\item The function $g$ is constant and equal to $M$, which is equivalent to $g^*(m) =1$ if $0<m\leq M$,
	and $g^*(m) = \infty$ otherwise.
\item The function  $g$ is strictly inflationary, i.e., $x < g(x)$ for every non-negative integer $x$,
	in which case  $g^*$ takes only finite values.
\end{enumerate}

The pseudo-code of the $\SAP$ algorithm is given below.
For the sake of simplicity, we have omitted the first rule so to compute $\clmod{C_i}{PM_i}$
	instead of just $C_i$: this rule which allows us to assume that $C_i(0) \leq PM_i(0) -1$
	is actually effective only in the very first round.

\begin{algorithm}[h]
\small
\begin{algorithmic}[1]
\REQUIRE{}
 \STATE $C_i \in  \IN$;
  \STATE $M_i \in  \IN$; 
\ENSURE{}
  \STATE send $\langle C_i, M_i \rangle$ to all \;
  \STATE receive $\langle C_{j_1}, M_{j_1} \rangle, \langle C_{j_2}, M_{j_2} \rangle, \dots$ from the set $\!\In_i$ of incoming neighbours \;
  \STATE $C_i \leftarrow \big [\underset{j \in \In_i}{\min}~C_j + 1 \big ] _{PM_i}$ \label{line:min} \;
  \STATE $M_i \leftarrow \underset{j \in \In_i}{\max}~M_j$ \label{line:max} \;
	\IF{$C_j\not\equiv_P C_{j'} $ for some $j, j' \in \!\In$}
	\STATE	$M_i \leftarrow g(M_i)$ \label{line:app_g}
  \ENDIF
 \end{algorithmic}
 \caption{The  $\SAP$ algorithm}
\label{algo:2:SAP}
\end{algorithm}

\subsection{Basic invariants }

We fix an execution of the algorithm $\SAP$ with the dynamic graph $\dG$.
For each round $t$ in this execution, let $\nmin{i}{t}$ denote any one of the $i$'s in-neighbor  in $\dG(t)$ 
	satisfying
	$$C_{\nmin{i}{t}} (t-1) = \min_{j\in \In_i(t) } C_j(t-1) .$$

The path $ i_0, i_1, \cdots ,  i_{\ell} $ in the round interval $ [ s , s + \ell -1 ] $  is said to be a \emph{synchronized path} 
	if for each index~$k \in \{ 0, \cdots, \ell -1 \}$, the pair $( i_k , i_{k+1}) $ is an edge in $\dG(s +k)$ with 
	$$ C_{i_{k+1}} (s + k ) \equiv_{P}  1 + C_{i_k} (s + k -1) .$$
Observe that the edge $(i_0,i_{\ell})$  in~$\dG( s : s + \ell -1 )$ may also correspond to another path 
	in $ [ s , s + \ell -1 ] $ that is non-synchronized.
The system is said to be \emph{synchronized in round} $t$ if 
	$$ \forall i,j \in V, \ \ C_i (t) \equiv_P C_j(t) .$$
	
We start with two preliminary lemmas.
The first one is a direct consequence of the code of $\SAP$, and its proof is omitted.
	
\begin{lem}\label{lem:general}
\begin{enumerate}
\item If the system is synchronized in round~$s$, then it is synchronized in any round~$t\geq~s$.
\item If $(i,j)$ is an edge in $\dG(s : t)$, then $C_j(t) \leq C_i(s - 1) + t -s +1 $.
\item Each variable $M_i$ is non-decreasing. 
\end{enumerate}
\end{lem}

\begin{lem}\label{lem:Ci}
For every  round $t\geq 1$ and every node $i\in V$, one of the following statements is true:
	\begin{enumerate}
		\item $C_i(t)  $ is positive and $C_i(t) = 1 + C_{\nmin{i}{t}} (t-1)$
		\item $C_i(t) =0$, $C_i(t-1) = PM_i(t-1) -1$, and $\nmin{i}{t} = i $. 
		\end{enumerate}
\end{lem}
\begin{proof}
The lemma just relies on the following series of inequalities:
	$$ C_{\nmin{i}{t}} (t-1) \leq C_i(t-1) \leq  PM_i (t-1) -1.$$
The last inequality is clear for $t=1$, and for $t\geq 2$, it is a consequence of $ C_i(t-1) \leq PM_i(t-2) -1 $
	and of the fact that $M_i$ is non-decreasing.
\end{proof}
\begin{lem}\label{lem:path}
If $(i,j)$ is an edge in $\dG(s : t)$, then one of the following statements is true:
	\begin{enumerate}
		\item $C_j(t) \equiv_P C_i(s - 1) + t - s +1 $;
		\item $M_j(t) \geq g(M_i(s-1))$.
	\end{enumerate}
\end{lem}

\begin{proof}
Let $ i= i_0, i_1, \cdots ,   i_{\ell} =j $, with $\ell = t-s+1$,  be a path in the round interval $ [ s , t ] $ corresponding to the 
	edge $(i,j)$ in $\dG(s : t)$.
There are two possible cases:
\begin{enumerate}
\item For each index~$k \in \{ 0, \cdots, \ell -1 \}$, it holds that
	$ C_{i_{k+1}} (s + k ) \equiv_{P}  1 + C_{i_k} (s + k -1) $, which implies 
	$$ C_j(t) \equiv_P C_i(s - 1) + t - s +1 .$$
\item Otherwise, let $i_{k +1}$ be the first  node in this path such that 
	$$ C_{i_{k+1}} (s + k ) \not\equiv_{P}  1 + C_{i_{k}} (s + k -1)  .$$
In round~$s+k$, the node $i_{k+1}$  receives the value $ C_{i_{k}} (s + k -1) $ from~$i_{k}$,
	and 
	it also receives  some value $c$ with  $c\not\equiv_P C_{i_{k}} (s + k -1) $. 
Then $i_{k+1}$ executes line~\ref{line:app_g}, which implies that $M_{i_{k+1}} (s +k) \geq g( M_{i_{k}} (s +k -1)$,
	and thus 
	$$ M_j(t) \geq M_{i_{k+1}} (s +k) \geq g( M_{i_{k}} (s +k -1)) \geq g(M_i(s-1)) .$$

\end{enumerate}
\end{proof}

\subsection{Strong connectivity with bounded delay}

We now determine some functions $g$ for which  the $\SAP$ algorithm achieves mod\,-$P$ synchronization 
	in the case of  strong connectivity with bounded delay. 
We fix such a dynamic graph~$\dG$ and an execution of $\SAP$ with $\dG$, 
	and we let $\diam(\dG)= D$.
	
\begin{lem} \label{lem:zerooustab}
For every $t  \in \IN$, one of the following statements is true:
	\begin{enumerate}
	\item there exist a node $i \in V $ and an integer  $ d \in \{1, \dots, D-1\}$  such that $C_i(t + d) = 0$;
	\item  the system is synchronized in round $t  +D$.
	\end{enumerate}
\end{lem}
\begin{proof}
Let us assume that all the counters $C_i$ are non-zero in the round interval $[t  + 1, t  +D-1]$.
First, we prove by induction on  $d$, $ 1 \leq d \leq D-1$, that
\begin{equation}
	\label{eq:minprod}
	\forall i \in V, \ \ \ C_i(t  +d) = d + \underset{j \in \In_i(t +1:t +d)}{\min}~C_j(t ) .
\end{equation}
\begin{enumerate}
\item The base case $d=1$ is an immediate consequence of Lemma~\ref{lem:Ci}.
\item Inductive step: let us assume that Eq.~(\ref{eq:minprod}) holds for some $d$ with $1 \leq d < D-1$.
For every node~$i$ in $V$, we have
\begin{align*}
	C_i(t +d+1) & = 1 + \min_{j\in \In_i(t +d+1)} C_{j} (t  + d) \\
			     &  = 1+ d + \underset{j \in \In_i(t +d+1)}{\min}~\left ( \underset{k \in \In_i(t +1:t +d)}{\min}~C_k(t ) \right )  \\
			     & =  1 + d  + \underset{k \in \In_i(t +1:t +d+1)}{\min}~C_k(t ).
\end{align*}
The first equality is a direct consequence of Lemma~\ref{lem:Ci}, the second one is by inductive hypothesis, 
	and the third one is due to the fact that
	$\dG(t +1:t +d+1) = \dG(t +1:t +d) \circ \dG(t +d+1)$.
\end{enumerate}
This completes the proof of Eq~(\ref{eq:minprod}) for every integer $d \in \{1\leq d \leq D-1\}$.

Then for each node $i$, we get 
\begin{align*}
	C_i(t +D) & = \left [ 1 + \underset{j \in \In_i(t +D)}{\min}~C_j(t  + D -1) \right ]_{PM_i(t +D-1)} \\
		  &  = \left [ D + \underset{j \in \In_i(t +D)}{\min}~\left ( \underset{k \in \In_i(t +1:t +D-1)}{\min}~C_k(t )\right )  \right ]_{PM_i(t +D-1)}  \\
			     & =  \left [ D + \underset{k \in V}{\min}~C_k(t )  \right ]_{PM_i(t +D-1)} .
\end{align*}
The second equality is  due to Eq~(\ref{eq:minprod}) at round $t +D-1$, and the third one is 
	a consequence of $\dG(t +1:t +D-1) \circ \dG(t +D) = \dG(t +1:t +D) = V$.
It follows that all the counters $C_i(t  +D)$ are equal modulo~$P$, i.e., 
	the system is  synchronized in round $t +D$.
\end{proof}

\begin{lem} \label{lem:propagc}
Let $t$ be a round in which $C_i(t) + D \leq PM_i(t)$ holds for each node $i$.
Then the system is synchronized in round $ t + D$.
\end{lem}
\begin{proof}
Let $i$ be any node, and let  $d \in \{1, \dots, D-1\}$.
We have
	\begin{align*}
		1 + \underset{j \in \In_i(t+d)}{\min}~C_j(t+d-1) & \leq  1 + C_i(t+d-1)\\
		&  <  D + C_i(t) \\
		& \leq PM_i(t) \\
		&\leq PM_i(t+d -1) .
	\end{align*}
The 	first inequality is due to the self-loop at node~$i$ in $\dG(t + d)$,
	the second and fourth ones are direct consequences of  the last two claims  in Lemma~\ref{lem:general},
	and the third inequality is the basic assumption of the lemma.
It follows that $C_i(t+d) \neq 0$, and Lemma~\ref{lem:zerooustab} shows that  
	the system is synchronized in round $t+D$.
\end{proof}

For any integer $t\in \IN$, we let
	$ M (t) \eqdef \min_{i\in V}  M_i(t) $.

\begin{lem} \label{lem:croissantsc}
For all non-negative integer $q \in \IN$, one of the following statements is true:
	\begin{enumerate}
		\item the system is synchronized in round $ q D $;
		\item $ M( q D) \geq g^q(M(0)) $.
	\end{enumerate}
\end{lem}
\begin{proof}
We proceed by induction on $q$.
The base case $q=0$ is trivial.
For the inductive step, assume that the lemma holds in round $ q D$, and
	consider the two following cases:
\begin{enumerate}
\item The system is synchronized in round $ q D $.
The first claim in Lemma~\ref{lem:general} asserts that the system remains synchronized in round $(q+1)D$.
\item Otherwise, we have $ M( q D) \geq g^q(M(0))$.
Let $j$ be a node that realizes $M((q+1)D)$, i.e.,  $M_j((q +1)D) = M((q +1)D)$.
Since the system is not synchronized in round~$qD$, there exists a node $i$  in $V$ such that 
	$$C_j((q+1)D) \not\equiv_P C_i(q D) +D .$$
Because $D$ is the diameter of $\dG$, $(i,j)$ is an edge of $\dG(q D+1 :  q D+D)$, 
	and we obtain:
	$$M_j((q+1)D )  \geq   g(M_i(q D))  \geq  g (M(q D))  \geq  g^{q +1}(M(0)). $$
The first inequality is by Lemma~\ref{lem:path}, and the last two ones  are due to the fact 
	that the function $g$  is non-decreasing. 
\end{enumerate}
\end{proof}

 \begin{thm} \label{thm:fcsync}
In any execution with a dynamic graph whose diameter $D$ is finite, 
	the $\SAP$ algorithm achieves mod-$P$ synchronization  for any non-decreasing function $g: \IN \rightarrow \IN$
	such that $ g^*\left ( \left \lceil \frac{2D}{P} \right \rceil \right ) $ is finite.
Moreover, the stabilization time is bounded by $ \left ( g^*\left ( \left \lceil \frac{2D}{P} \right \rceil \right )  + 2 \right )D$.
\end{thm}
\begin{proof}
We let $q_0 = g^*\left ( \left \lceil \frac{2D}{P} \right \rceil \right ) $;
	hence $q_0 \geq 1$.
By Lemma~\ref{lem:zerooustab}, either the system is synchronized in round $q_0 \, D $, or 
	there exist a node $i_0$ and an  integer  $ d \in \{1, \dots,  D - 1\}$ such that 
	$C_{i_0}(q_0 \, D + d -D) = 0$.
In that case, we have:
	$$ M( q_0 \, D) \geq g^{q_0} (M(0)) \geq \frac{2D}{P} .$$ 
The first inequality is Lemma~\ref{lem:croissantsc} and the second one is due to   $M(0) \geq 0$,
	$g$ is non-decreasing, and the definition of $q_0$.
	
Moreover, the digraph $\dG(q_0 \, D  -D +d + 1 : q_0 \, D + d  )$ is complete since $D$ 
	is the diameter of $\dG$, 
	and the second claim in Lemma~\ref{lem:general} shows that for every node $i $, we have
	$$ C_i (q_0 \, D + d ) \leq  D  .$$
Hence,
	\begin{align*}
		PM_i (q_0  D +d ) &\geq PM_i(q_0 \, D) \\
		&\geq PM(q_0 \, D) \\
		&\geq 2D  \\
		&\geq C_i (q_0 \,  D + d) + D
		\end{align*}
Finally, Lemma~\ref{lem:propagc} shows that the system is synchronized in round $(q_0+1)D + d$.		
\end{proof}

\subsection{Specializations of the \emph{SAP} algorithm}

Theorem~\ref{thm:fcsync} leads to two corollaries corresponding to two strategies for the choice of $g$.
Firstly, when some bound $B$ on the diameter of the dynamic graph is given, we may choose 
	$g$ to be the constant function $\lambda x.M$ where $M= \left \lceil \frac{2  B }{P}\right \rceil$.
Then we get $q_0 = 1$ and the pseudo-code of the algorithm $\FP$ may be simplified accordingly
	(cf. Algorithm~\ref{algo:2:FP}).

\begin{algorithm}[h]
\small
\begin{algorithmic}[1]
\REQUIRE{}
 \STATE $C_i \in  \IN$;
\ENSURE{}
  \STATE send $\langle C_i  \rangle$ to all \;
  \STATE receive $\langle C_{j_1} \rangle, \langle  C_{j_2} \rangle, \dots$ from the set $\!\In_i$ of in-neighbors \;
  \STATE $C_i \leftarrow \big [\underset{j \in S}{\min}~C_j + 1 \big ] _{PM}$ \label{line:minfp} \;
 \end{algorithmic}
	\caption{The  $\FP$ algorithm}
\label{algo:2:FP}
\end{algorithm}

\begin{cor}\label{cor:FP}
The $\FP$ algorithm solves the mod-$P$ synchronization  problem in any dynamic graph with a diameter
	less or equal to $PM/2$. 
\end{cor}
Let us observe that Theorem~\ref{thm:fcsync} provides an upper-bound of three times the diameter~$D$ 
	on $\FP$'s stabilization time, which is independent on the bound $B$.

The limit of $PM/2$ in Corollary \ref{cor:FP} is tight, as proved  by the following $\FP$'s execution.
For simplicity, we assume that $P$ is even and $M$ is odd.
Let $n$ be an integer such that $ n > \frac{PM}{2}+1 $.
The communication graph is the static bidirectional chain $i_0, \cdots, i_{n-1}$.
The node $i_0$ starts with~0, whereas the other nodes start with $PM/2$.
We can prove that in round~$t$, there are exactly $\frac{PM}{2} + 1 - \left | \clmod{t}{PM} - \frac{PM}{2} \right |$ nodes 
	with a local clock equal to $\clmod{t}{PM}$, and all the other clocks are equal to 
	 $\clmod{t + \frac{PM}{2}}{PM}$. 
Since $M$ is odd, we have $t + \frac{PM}{2} \not\equiv_P t$. 
As it holds that 
	$$  1 \leq \frac{PM}{2} + 1 - \left | \clmod{t}{PM} - \frac{PM}{2} \right | \leq \frac{PM}{2} + 1  < n , $$
	both values are present, which shows that the system never synchronizes.
	
Interestingly, the self-stabilizing algorithm in~\cite{BPV:algorithmica:08}, called \emph{SS-MinSU} and developed 
	for  clock synchronization in a static and strongly connected network when a bound $B$ on the diameter\footnote{%
	The bound $B$  is  denoted~$\alpha$ in  the \emph{SS-MinSU} algorithm.}
	 is available, is actually  an optimization of the $\FP$ algorithm:
	the use of negative numbers in the interval $[-B, -1]$ allows for reducing the number of states to 
	$B +P$ instead of $ \left \lceil \frac{2 B}{P} \right \rceil P$ in $\FP$ (see Figure~\ref{fig:space}).
	
As for the algorithm proposed in~\cite{ADG:ppl:91} for a connected  bidirectional digraph~$G$, it corresponds 
	to the $SAP_{\lambda x.1}$  algorithm, combined with a round-robin strategy which consists, for each node, to send one 
	message per round according to this fixed cyclic order amongst the out-neighbors in $G$.
This strategy thus translates the  digraph~$G$ into a dynamic graph $\dG$.
Moreover, if $ i = k_1 , \dots , k_{m+1}=j$  is a path  in the digraph~$G$, then $(i,j)$ is an edge in any 
	digraph $\dG \left(t +1 : t+ d^-_{k_1} + \dots + d^-_{k_m}\right ) $, where $ d^-_{k}$ denotes $k$'s out-degree
	in~$G$.
In the case~$G$ is bidirectional, Proposition 24 in~\cite{CB:infcomp:22} shows that
	$$ d^-_{k_1} + \dots + d^-_{k_m} \leq 3 |V| $$
	if the path $k_1 , \dots , k_{m} $ is a geodesic in~$G$.
In other words, the dynamic graph~$\dG$ has a finite diameter which  is upper-bounded by $3|V|$.
Via Corollary~\ref{cor:FP}, the interpretation of the algorithm in~\cite{ADG:ppl:91} for a (fixed) bidirectional  digraph 
	in terms of a run of $SAP_{\lambda x.1}$ over a dynamic graph
	shows that this algorithm works when $P \geq 6 |V|$, and its stabilization time is less than $9 |V|$
	 (instead of the correctness condition $P \geq n^2$ and the stabilization bound of $\frac{3}{2} n^2$ 
	 given  both in~\cite{ADG:ppl:91}).

\vspace{0.2cm}
		
When the diameter of the dynamic graph is finite but  no bound is available,
	we may use the following corollary of Theorem~\ref{thm:fcsync}: 
\begin{cor}\label{cor:SAP}
For any non-decreasing and inflationary function~$g$,
	the $\SAP$ algorithm solves the mod-$P$ synchronization  problem in any dynamic graph 
	that is strongly connected with bounded delay.
\end{cor}

Our $\SAP$ algorithm is a variant of the algorithm presented by Boldi and Vigna~\cite{BV:dc:02}:
	both rely on the  idea of a self-adaptive period and  their time complexities are of
	the same order of magnitude.
The main discrepancy lies in space complexity: while the variables $C_i$ in $\SAP$ are of the order of
	$PM(q_0 \, D)$, the algorithm in~\cite{BV:dc:02} uses variables
	of the order of $PM(q_0 \, D)^2$,
	where $q_0 = g^*\left ( \left \lceil \frac{2D}{P} \right \rceil \right )$.

\section{Periodic  Clocks with Uniform Rootedness}

The aim of this section is to study how the assumption of strong connectivity with bounded delay
	(or equivalently of a finite diameter) can be relaxed 
	so that the $\SAP$ algorithm still achieves mod-$P$ synchronization.

\subsection{The \emph{SAP} algorithm with rootedness}\label{sec:cex}

We first demonstrate that, as opposed to the $\MM$  algorithm, 
	the sole assumption of a non-empty center is not sufficient for $\SAP$ to synchronize nodes.
Indeed, even with a central node~$i$, sporadic roots may disrupt the value of  $i$'s clock,
	and hence preclude any alignment of the other clocks on $C_i$.
This is  the idea underlying the scenario that we develop below:
Let $G, H_j, H_k, I$ be the four digraphs defined  in Figure~\ref{fig:rootedSAP} 
	with three nodes $ i,j,k $, and let $\Phi_k$ be the following predicate on the rounds
	of a $\SAP$ execution:
	$$ \big ( M_i = M_j \big ) \wedge \big (M_i \geq M_k \big )
	     \wedge \big ( C_i = C_j  \big ) \wedge  \big ( C_i\not\equiv_P 0 \big )   \wedge  \big ( C_i \leq PM_i -2 \big ) 
	    \wedge   \big ( C_k = 0 \big )  .$$
The predicate $\Phi_j$ is obtained by exchanging  $j$ and $k$.
Then we easily prove the following lemma:
\begin{lem}\label{lem:cex1}
Let $t$ be a round of a $\SAP$ execution  with a dynamic graph $\dG$ such that 
	$$\dG(t+1) = \cdots \dG(t +  PM - c -2) = G, \  
	 \dG(t +  PM - c -1) = H_k, \ \dG(t +  PM - c) = I .$$
If $\Phi_k$  holds at round $t$, then $\Phi_j$  holds at round $t + PM- c $,
	where $M = M_i (t)$ and $c= C_i(t)$.
Moreover, $M_i( t + PM- c ) = g^{PM -c -1} (M)$ and $C_i(t + PM- c ) = PM -c $.
\end{lem}
We now fix two positive integers $M^0$ and $c^0$ such that 
	$ c^0 \in \{ 1, \cdots, PM^0 -2\}$  and $ c^0 \not\equiv_P 0 $,
	and we consider the two sequences $(M^r)_{r\geq 0}$ and $(c^r)_{r\geq 0}$ 
	satisfying
	$$\left \{ \begin{array}{l}
		M^{r+1} = g^{PM^r -c^r -1} (M^r)\\
		c^{r+1} = PM^r - c^r. 
		\end{array} \right. $$
We let $M^{-1} = 0$.
The dynamic graph $\dG$, defined as:
	$$ \dG(PM^{r-1} +1) = \cdots = \dG(PM^r - c^r -2) = G, \ 
	    \dG(PM^r - c^r -1) = H_k \mbox{ or } H_j , \ 
	    \dG(PM^r - c^r -1) = I ,$$
	    is rooted with delay two and $i$ is its unique center.
Lemma~\ref{lem:cex1} shows that $\Phi_k$ holds infinitely often in the $\SAP$ execution	with 
	the dynamic graph~$\dG$ and starting with:
	$$ M_i(0) = M_j(0) = M_k(0) = M^0\!,\   C_i(0) = C_j(0)  = c^0, \ \mbox{ and } \  C_k(0) = 0 ,$$
	which proves that the nodes are never synchronized.

\begin{figure}
    \begin{subfigure}[b]{0.22\textwidth}
    \centering
        \resizebox{\linewidth}{!}{
			\begin{tikzpicture}[roundnode/.style={circle, draw=green!60, fill=green!5, very thick, minimum size=7mm}]
				\node[roundnode]        (first)        {$i$};
				\node[roundnode]        (second)       [right=of first] {$j$};
				\node[roundnode]        (third)       [below=of first] {$k$};

				\draw[-stealth] (first.east) -- (second.west);
				\draw[-stealth] (first.south) -- (third.north);
				\draw[-stealth] (first.north) .. controls +(north:9mm) and +(west:9mm) .. (first.west);
				\draw[-stealth] (second.north) .. controls +(north:9mm) and +(right:9mm) .. (second.east);
				\draw[-stealth] (third.south) .. controls +(south:9mm) and +(west:9mm) .. (third.west);
            \end{tikzpicture}
        }
        \caption{digraph $G$}   
    \end{subfigure}
     \begin{subfigure}[b]{0.22\textwidth}
        \centering
        \resizebox{\linewidth}{!}{
			\begin{tikzpicture}[roundnode/.style={circle, draw=green!60, fill=green!5, very thick, minimum size=7mm}]
				\node[roundnode]        (first)        {$i$};
				\node[roundnode]        (second)       [right=of first] {$j$};
				\node[roundnode]        (third)       [below=of first] {$k$};

				\draw[stealth-stealth] (first.east) -- (second.west);
				\draw[-stealth] (first.south) -- (third.north);
				\draw[-stealth] (first.north) .. controls +(north:9mm) and +(west:9mm) .. (first.west);
				\draw[-stealth] (second.north) .. controls +(north:9mm) and +(right:9mm) .. (second.east);
				\draw[-stealth] (third.south) .. controls +(south:9mm) and +(west:9mm) .. (third.west);
            \end{tikzpicture}
        }
        \caption{digraph $H_j$}
        \label{fig:gr2}
    \end{subfigure}
    \begin{subfigure}[b]{0.22\textwidth}
        \centering
        \resizebox{\linewidth}{!}{
			\begin{tikzpicture}[roundnode/.style={circle, draw=green!60, fill=green!5, very thick, minimum size=7mm}]
				\node[roundnode]        (first)        {$i$};
				\node[roundnode]        (second)       [right=of first] {$j$};
				\node[roundnode]        (third)       [below=of first] {$k$};

				\draw[-stealth] (first.east) -- (second.west);
				\draw[stealth-stealth] (first.south) -- (third.north);
				\draw[-stealth] (first.north) .. controls +(north:9mm) and +(west:9mm) .. (first.west);
				\draw[-stealth] (second.north) .. controls +(north:9mm) and +(right:9mm) .. (second.east);
				\draw[-stealth] (third.south) .. controls +(south:9mm) and +(west:9mm) .. (third.west);
            \end{tikzpicture}
        }
        \caption{digraph $H_k$}
        \label{fig:gr3}
    \end{subfigure}
    \begin{subfigure}[b]{0.22\textwidth}
        \centering
        \resizebox{\linewidth}{!}{
			\begin{tikzpicture}[roundnode/.style={circle, draw=green!60, fill=green!5, very thick, minimum size=7mm}]
				\node[roundnode]        (first)        {$i$};
				\node[roundnode]        (second)       [right=of first] {$j$};
				\node[roundnode]        (third)       [below=of first] {$k$};

				\draw[-stealth] (first.north) .. controls +(north:9mm) and +(west:9mm) .. (first.west);
				\draw[-stealth] (second.north) .. controls +(north:9mm) and +(right:9mm) .. (second.east);
				\draw[-stealth] (third.south) .. controls +(south:9mm) and +(west:9mm) .. (third.west);
            \end{tikzpicture}
        }
        \caption{digraph $I$}
        \label{fig:gr1}
    \end{subfigure}
	\caption{Four digraphs with three nodes.}\label{fig:rootedSAP}

\end{figure}
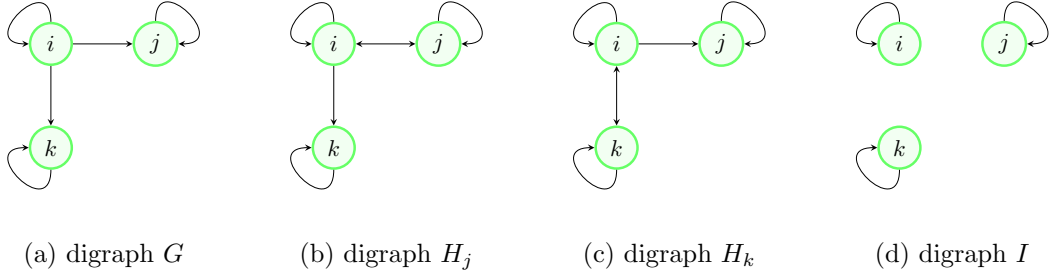

That leads us to consider the stronger assumption of \emph{uniform} rootedness.
However, the simple but typical scenario below shows that
	the correctness proof of $\SAP$ that we have developed in the previous section
	cannot be directly extended to dynamic graphs that are uniformly rooted with bounded delay:
	the $\FP$ algorithm, 
	does not achieve mod-$P$ synchronization in the execution  with the initial values 
	$C_i(0) = C_j(0) = 1$ and $C_k(0) = 0$ and the fixed graph $H$ defined in Figure~\ref{fig:cexFP},
	even for large value of $M$.
Indeed, at each round $t$, it holds that $C_i(t) =\clmod{t+1}{PM}$, $C_k(t) =\clmod{t}{PM}$, and
$$C_j(t) = 
\begin{cases}
	1 & \text{if}~\clmod{t}{PM} = 0 \\
	\clmod{t}{PM} & \text{otherwise.}
\end{cases}$$

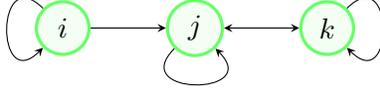
\begin{figure}[htb]
	\centering
	\begin{tikzpicture}[roundnode/.style={circle, draw=green!60, fill=green!5, very thick, minimum size=7mm}]
		\node[roundnode]        (first)        {$i$};
		\node[roundnode]        (second)       [right=of first] {$j$};
		\node[roundnode]        (third)       [right=of second] {$k$};

		\draw[stealth-stealth] (second.east) -- (third.west);
		\draw[-stealth] (first.east) -- (second.west);
		\draw[-stealth] (first.north west) .. controls +(north west:9mm) and +(south west:9mm) .. (first.south west);
		\draw[-stealth] (third.north east) .. controls +(north east:9mm) and +(south east:9mm) .. (third.south east);
		\draw[-stealth] (second.south west) .. controls +(south west:9mm) and +(south east:9mm) .. (second.south east);
	\end{tikzpicture}
	\caption{The  digraph $H$ with three nodes. }\label{fig:cexFP}
\end{figure}

\subsection{The  \emph{SAP} algorithm and uniform rootedness with bounded delay}
	 
The striking point of increasing periods is precisely to overcome the above-mentioned limitation:
	we are going to prove that the $\SAP$ algorithm achieves mod\,-$P$ synchronization 
	in the case of  uniform rootedness with bounded delay under the sole condition
	of a non-decreasing and strictly inflationary function~$g$.
In other words, while Corollary~\ref{cor:FP} has no counterpart for dynamic graphs
	that are uniformly rooted with bounded delay, we will show that 
	Corollary~\ref{cor:SAP} extends to  this latter class of dynamic graphs,
	with a synchronization phenomena quite different from that involved in the case 
	of strong connectivity.

We fix a dynamic graph~$\dG$ that, ultimately,  is uniformly rooted with bounded delay,
	 and an execution~$\sigma$ of $\SAP$ with $\dG$.
Without loss of generality, we may assume that $\dG$ is
	 uniformly rooted with bounded delay from the beginning, and we  let $Z = \Z(\dG)$ and $ R =\rad(\dG)$.
	
The nodes in $Z$ receive no message from the nodes in  $V\setminus Z$.
From the viewpoint of every node in~$Z$, the execution~$\sigma$ is thus indistinguishable
	from an execution with the set of nodes equal to~$Z$ and a dynamic graph
	that is strongly connected with bounded delay.
Theorem~\ref{thm:fcsync} shows that mod~$\!P$-synchronization is eventually achieved in~$Z$.
A closer look at its proof yields the following more precise result:
	there exist two non-negative integers $s$ and $\MZ$ such that 
	\begin{equation}\label{eq:Zsynchro}
	 \forall t\geq s, \ \forall k,\ell \in Z :  \  C_k(t) = C_{\ell}(t) \ \mbox{ and } \  M_k(t) = \MZ .
	 \end{equation}
The minimum integer $s$ satisfying Eq.~(\ref{eq:Zsynchro}) is denoted by $t_0$, and 
	$C(t)$ is the common value of all the  counters $C_k$(t) for $k \in Z$ and $t\geq t_0$.
The node~$i$ is said to be \emph{$Z$-synchronized at round}~$t$ if
	$$  C_i(t) \equiv_{P} C(t) .$$
The set of $Z$-synchronized nodes at round $t$ is denoted by~$S_Z(t)$.
In the case the system is not synchronized in round~$t$, i.e.,  $S_Z(t) \neq V$,  we let
	$$ \tilde{M} (t) \eqdef  \min_{i\notin S_Z(t)}  M_i(t) .$$

Using the existence of a self-loop at each node and the update rules of the variables $M_i$, 
	we easily show that~$\tilde{M}$ is non-decreasing:
	
\begin{lem} \label{lem:Mtilde}
For all $t\geq t_0$, it holds that $ \tilde{M} (t + 1) \geq \tilde{M} (t) $.
\end{lem}

Let $i$ be a central node such that $ \e_{\dG}(i) =R$, and let $j$ be an arbitrary node.
There exists an edge $(i,j)$ in each digraph $\mathds{G}(t :t+R-1)$.
Since $C_i(t) < PM$, the second claim in Lemma~\ref{lem:general} implies the following upper bound 
	on the clock~$C_j$:
	
\begin{lem}\label{lem:Cuppbound}
For all $t\geq t_0 +R $ and all nodes $j \in V$, it holds that $ C_j (t) <  P\MZ + R $.
\end{lem}

Then  Lemma~\ref{lem:croissantsc} admits the following  counterpart in the case of uniform rootedness.

\begin{lem} \label{lem:Mlowbound}
For every positive integer $q$, 
	one of the following statements is true:
	\begin{enumerate}
		\item the system is synchronized in round $t_0 + q R$; 
		\item $\tilde{M}(t_0 + q R) \geq g^{q-1}(\MZ)$.
	\end{enumerate}
\end{lem}
\begin{proof}
We proceed by induction on $q$.
The base case $q=1$ is due to the update rule for~$M_i$.
For the inductive step, assume that the lemma holds in round $t_0 + q R$ and
	that some node~$i$ is not $Z$-synchronized in round~$ t_0 + (q + 1) R $.
Since $ R = \rad(\dG)$, the node $i$ has an in-neighbor in~$Z$
	in the directed graph~$\dG(t_0 + q R+1 : t_0 + (q + 1)R)$, i.e., there exist a node $j\in Z$ and a path
	$ j= j_0, j_1, \cdots ,  j_R =i$  in the round interval $ [ t_0 + q R+1 , t_0 + (q + 1) R ] $. 
Since $j\in Z \subseteq S_Z(t_0 + q  R)$ and $i \notin S_Z(t_0 + (q + 1) R )$, this path is not synchronized.
Let~$d\in\{1,\dots,R\}$ be the first  index  such that $j_{d -1}\in S_Z(t_0 + q R + d -1)$ and $j_{d} \notin S_Z(t_0 + q R + d)$.
	Then $j_{d-1}$ and $\nmin{(j_d)}{t_0+qR+d}$ (denoted $j^+$ for short) are two in-neighbors of $j_d$ whose clocks are not congruent 
	modulo~$P$ in round $t_0+qR+d-1$.
It follows that:
		\begin{equation*}
		M_i(t_0 + (q + 1) R ) \geq M_{j_{d}} (t_0 + q  R + d ) \geq g( M_{j^+} (t_0 + q  R +d -1)) \geq g(\tilde{M}( t_0 +  q R + d-1))  
		\geq g^{q }(\MZ).
	\end{equation*}
	The first two inequalities 	are due to the update rules for $M_i$ and $M_{j_{d}}$, 
	the third one is by definition of $\tilde{M}$ and the fact that $g$ is non-decreasing,
	and the last one is a consequence of  the inductive assumption and Lemma~\ref{lem:Mtilde}.
\end{proof}

 \begin{thm} \label{thm:general}
For any non-decreasing and inflationary function~$g$,
	the $\SAP$ algorithm solves the mod-$P$ synchronization  problem in any dynamic graph 
	that, ultimately, is uniformly rooted with bounded delay.
\end{thm}
\begin{proof}
We let $q_1 = g^*\left ( \left \lceil \MZ + \frac{R+1}{P} \right \rceil \right ) $
	and  $t_1 = t_0  + q_1  R$.
The main part of the proof consists in showing, by induction on $t\geq t_1$, the following property:
	\begin{equation*}
		\forall i\in V\setminus S_Z(t), \ \ C_i(t) \geq t-t_1 .
	\end{equation*}
The base case $t = t_1$ is obvious. 
Suppose now that the above property holds at round~$t \geq t_1$,
	 and that $S_Z(t+1) \neq V$.
Let us fix some node $j \notin S_Z(t+1)$; we are going to show that
	\begin{equation}\label{eq:lemC}
	C_j(t+1) \geq t +1 -t_1 .
	\end{equation}
Lemma~\ref{lem:Ci} shows that either (a) $\!C_j(t+1) \! = \!1 + C_{\nmin{j}{t+1}}\!\!(t) \!$ or 
	(b) $\!C_j(t) \! = \! PM_j(t) - 1$.
In case~(a), the inequality~(\ref{eq:lemC}) follows from the inductive assumption 
	since  $\nmin{j}{t+1} \notin S_Z(t)$.
In case (b), Lemma~\ref{lem:Cuppbound} implies that $P M_j(t) = C_j(t) + 1 < P \MZ + R +1 $
	since $t\geq t_1 \geq t_0 + R$ . 
Moreover, it holds that 
	\begin{equation*}
	M_j(t) \geq \tilde{M}(t) \geq \tilde{M}(t_1) \geq g^{q_1} (\MZ) \geq  g^{q_1} (0) \geq   \MZ + \frac{R+1}{P} .
	\end{equation*}
The first inequality is due to the fact that $ j= j^+(t) \notin S_Z(t)$,  the second  one holds by
	Lemma~\ref{lem:Mtilde}, and the third one  is a consequence of Lemma~\ref{lem:Mlowbound}.
The forth inequality is due to the fact that $g$ is non-decreasing, and the last one is by 
	definition of $q_1$.
Thus case (b) cannot occur, which completes the proof of~(\ref{eq:lemC}).

To complete the proof of the theorem, we proceed by  contradiction, 
	and we  assume that some node $i$ is not $Z$-synchronized 
	in round~$t_2 = t_1 + P\MZ + R$.
Then, we would obtain both 
	$ C_i(t_2) \geq  P\MZ + R$ by the inequality (\ref{eq:lemC})
	and $ C_i(t_2) <  P\MZ + R$ by Lemma~\ref{lem:Cuppbound} since $t_2\geq t_0 +R$.
It follows that all the nodes are $Z$-synchronized in round $t_2$.
\end{proof}

%
%
%
%
%
%
%
%

\section{Concluding Remarks}

The correctness proofs presented above also provide bounds on stabilization time and 
	space complexity of the $\MM$ and  $\SAP$ algorithms;
	see Figures~\ref{fig:time} and~\ref{fig:space}. 
In the case where $g$ is a non-decreasing and inflationary function, the time bound and the space bound for $\SAP$ 
	depend respectively  on the functions~$g^*$ and $g$, leading thus to a time-space trade-off for choosing~$g$: 
	the more inflationary $g$ is, lower the time complexity is, and higher its space complexity is.
In both tables, we have included the complexity results of the $SynchMod_{_P}$ algorithm~\cite{PMCBM:sss:21} that
	solves the \emph{mod-$P$ Firing Squad} problem -- a problem which is similar  to mod-$P$ synchronization 
	with the additional requirement that each node is aware when synchronization is reached --
	in a dynamic networked system with a finite radius and a given bound on the radius.
For a more precise comparison with the algorithms presented in this paper, it is important 
	to note that the  $SynchMod_{_P}$ algorithm works only in the model of asynchronous starts.
	
In the case of infinite memory, the self-stabilizing $\MM$ algorithm is a very robust solution which works 
	with highly dynamic communication graphs and under a weak connectivity assumption, namely rootedness.
With finite memory and the $\SAP$ algorithm,  this assumption has to be strengthened into uniform rootedness.
In both cases,  these connectivity properties have to hold in bounded periods of time.
Thus, these results highlight  the critical importance of timing bounds for the network to be connected enough,
	and  demonstrate how time may act as a healer.
Moreover, as exemplified by the $\FP$ and $SynchMod_{_P}$ algorithms, the initial knowledge on these 
	bounds allows for much more efficient solutions in terms of both time and space.
	
Concerning connectivity assumptions, the first scenario in Section~\ref{sec:cex} shows
	that the $\SAP$ algorithm does not work anymore when relaxing the assumption of 
	uniform rootedness  into the one of (simple) rootedness, even in the case of a non-empty center.
A natural question then arises about the possibility of designing a finite-state self-stabilizing algorithm
	that provides nodes with clocks modulo $P$ which eventually synchronize in a dynamic graph
	with a finite radius. 

\begin{figure}
\small{
\begin{tabular}{|C{3.8cm}|C{2.7cm}|C{1.7cm}|C{3.4cm}|C{2.2cm}|}
\hline  Assumptions &	$\MM$ & $SAP_{\lambda x.\left\lceil\frac{2B}{P}\right\rceil}$ & $\SAP$ & $SynchMod_{\left\lceil\frac{B}{P}\right\rceil}$ \\
\hline  $\diam(\dG)= D \leq B$   & \multirow{2}*{$2 D + h(0) $} & $3 D $ & \multirow{2}*{$\left (2 + g^*\left ( \left\lceil \frac{2D}{P}\right\rceil \right) \right )D$} 
												   &  $  4 \,P   \left\lceil \frac{B}{P}\right\rceil $\\
\cline{1-1} \cline{3-3} \cline{5-5} $\diam(\dG)= D < \infty $   & 	 & --	&	& --	\\
\hline
uniformly rooted  ($Z$) &        & \multirow{2}*{ -- }&      &   \multirow{2}*{$ 6\, P |V|  \left\lceil \frac{B}{P}\right\rceil $}\\
$\rad(\dG) =R \leq B$ & $ 2 D\!+\!  2 R\!+\! h(0) $ &  & $\!R\! \left(1\!+\! g^*\!\!\left ( M\! +\! \left\lceil  \frac{2+R}{P}\right\rceil \right ) \right) +$ &  \\
\cline{1-1} \cline{3-3} \cline{5-5} 
uniformly rooted   ($Z$) & \tiny{with $h(0) = \max h_i(0)$}	 & \multirow{2}*{--}	& $ PM + \!(2 \!+\!g^*\!\left ( \left\lceil \frac{2D}{P}\right\rceil \right) D$ 	& \multirow{2}*{--} \\
$\rad(\dG)\! =\!R ,\diam(Z) \!= \!D$  &     & & \multirow{2}*{\tiny{with $\left\{\begin{array}{l}
                                                                                                           \!\!\! M = g^{ T}(\max M_i(0) )\\
                                                                                                          \! \!\! T = (2 + g^*\left ( \left\lceil \frac{2D}{P}\right\rceil \right) D
                                                                                                            \end{array}\right.$}} &  \\
   &  & &  &  \\
\hline								 
\end{tabular}} 
\caption{Bounds on stabilization time of four  clocks.}\label{fig:time}		
\end{figure}

\begin{figure}	
\small{
\begin{tabular}{|C{3.8cm}|C{1.7cm}|C{1.7cm}|C{3.4cm}|C{2.2cm}|}
\hline  Assumptions &	$\MM$ & $SAP_{\lambda x.\left\lceil\frac{2B}{P}\right\rceil}$ & $\SAP$ & $SynchMod_{\left\lceil\frac{B}{P}\right\rceil}$ \\
\hline  
\multirow{2}*{$\diam(\dG)= D \leq B$ }  &  \multirow{4}*{ }	 & \multirow{2}*{$ \left\lceil\frac{2B}{P}\right\rceil P$ } & \multirow{2}*{  } &  
				\multirow{2}*{$ B $ } \\
	&$\infty$	&	& $(P+1) g^{ T}(\max M_i(0) )$	&	\\										
\cline{1-1} \cline{3-3} \cline{5-5} 
\multirow{2}*{$\diam(\dG)= D < \infty $ }  & 	 & \multirow{2}*{--} & \tiny{ with $T = (2 + g^*\left ( \left\lceil \frac{2D}{P}\right\rceil \right) D $}  & --	\\
     &  &  &  &  \\
\hline								 
\end{tabular}} 
\caption{Memory bounds of four clocks (in the case of a  finite diameter).}\label{fig:space}	
\end{figure}

\paragraph{Acknowledgements:} We would like to thank Stephan Merz, Patrick Lambein-Monette,
	and Guillaume Pr\'emel  for very useful discussions.
We are also indebted to  Paolo Boldi and Sebastiano Vigna for their very deep and inspiring work on self-stabilization.

\bibliographystyle{plain}

\end{document}